\documentclass[12pt]{article}
\usepackage{amssymb, amsthm,amsmath,graphics,graphicx}

\newtheorem{theorem}{Theorem}
\newtheorem{acknowledgement}[theorem]{Acknowledgement}

\newtheorem{example}[theorem]{Example}

\newtheorem{lemma}[theorem]{Lemma}

\newtheorem{proposition}[theorem]{Proposition}
\newtheorem{remark}[theorem]{Remark}

\title{Dynamical Properties of Gaussian Chains and Rings with Long Range Interactions}
\author{
Wolfgang Bock\\
{\small Technomathematics Group}\\
{\small	University of Kaiserslautern}\\
{\small	P.\ O.\ Box 3049, 67653 Kaiserslautern, Germany}\\
{\small E-Mail:bock@mathemaik.uni-kl.de}\\[.3cm]
Jinky B. Bornales\\
{\small MSU-IIT Iligan, Andres Bonifacio Avenue, Tibanga},\\ 
{\small	9200 Iligan City, Philippines}\\
{\small E-Mail:jinky.bornales@g.msuiit.edu.ph}\\[.3cm]
Ludwig Streit\\
{\small CIMA, University of Madeira, Campus da Penteada},\\
{\small	9020-105 Funchal, Portugal}\\
{\small BiBoS, Universit{ä}t Bielefeld, Germany}\\
{\small E-Mail:streit@uma.pt}}

\begin{document}

%\preprint{}




%\thanks{}



%\keywords{}

\maketitle

\begin{abstract}
	Various authors have invoked discretized fractional Brownian (fBm) motion as
	a model for chain polymers with long range interaction of monomers along the
	chain. We show that for these, in contrast to the Brownian case, linear
	forces are acting between \textit{all }pairs of constituents, attractive for
	small Hurst index $H$ and mostly repulsive when $H$ is larger than 1/2. In
	the second part of this paper we extend this study to periodic fBm and
	related models with a view to ring polymers with long range interactions.
\end{abstract}

\section{Introduction}

Fractional Brownian motion (fBm) has, among other applications, attracted
some interest as a model for the conformation of chain polymers \cite{BBCES}%
\cite{Cher}\cite{GOSS}\cite{Hammouda}, as a generalization of the classical
Brownian models (see e.g.\ \cite{RC}, and references therein).\
Mathematically this enlarged class involves non-Martingale, non-Markovian
processes, physically on the other hand it is of interest because of its
more general scaling properties and the long range interactions it describes.

Fractional Brownian motion $B^{H},\,\ $conventionally\ starting at $%
B^{H}(0)=0,$ gives rise to a discrete, mean zero Gaussian process via%
$$
	X_{k}=B^{H}(k),\quad k=0,1,2,....
$$%
We denote the increments as%
$$
	Y_{k}=X_{k+1\ }-X_{k\ },\quad k=0,1,2,...,n-1
$$%
with finite dimensional probability densities $\rho _{n}$ given by%
$$
	\int d^{n}y\rho _{n}(y)e^{i(y,\lambda )}=\mathbb{E}\left( e^{i(Y,\lambda
		)}\right) =\exp \left( -\frac{1}{2}\left( \lambda ,R\lambda \right) \right)
$$%
and the covariance matrix%
$$
	R_{ik}=\mathbb{E}\left( Y_{i}Y_{k}\right) =\frac{1}{2}\left\vert
	k-l+1\right\vert ^{2H}+\frac{1}{2}\left\vert k-l-1\right\vert
	^{2H}-\left\vert k-l\right\vert ^{2H}.
$$
Hence%
$$
	\rho _{n}(y)=const.\exp \left( -\frac{1}{2}\left( y,Ay\right) \right)\, \text{
		with  }A=R^{-1},
$$%
which can be interpreted as a canonical density with regard to an energy
given by $\left( y,Ay\right) .$

The first part of this paper intends to elucidate the nature of the
corresponding interaction between monomers. We show that these (and more
general Gaussian) models can be interpreted as equilibrium models for
particles with attractive and repulsive elastic forces acting between any
two of them, i.e.%
$$
	\left( y,Ay\right) =\sum g_{kl}(x_{k}-x_{l})^{2}
$$%
and we display the strength $g_{kl}$ of these forces for different Hurst
indices $H$.

Periodic fBm cannot be constructed \ in a direct analogy to the Brownian
case as in (\ref{pBm}) because of the non-Markovian nature of fBm. As
pointed out by Istas \cite{Istas}, this problem can be solved by replacing
the distance $D$ in%
$$
	E\left( \left( B^{H}(t)-B^{H}(s)\right) ^{2}\right) =D^{2H}(t-s)=\left\vert
	t-s\right\vert ^{2H}
$$%
by the geodesic distance $d$ on the circle of length $L:$%
$$
	d(t-s)=\min \left( \left\vert t-s\right\vert ,L-\left\vert t-s\right\vert
	\right) ,
$$%
however with the limitation that this leads to a positive semi-definite
covariance function, and hence a Gaussian process. This is only the case if $%
H\leq 1/2.$ For this class of models we shall in particular compute the
coupling constants $g_{kl}$ and the \ spectrum of the quadratic form $\sum
g_{kl}(x_{k}-x_{l})^{2}$ for different values of the Hurst parameter $H$.

Finally we propose a class of periodic models which, contrary to the
non-admissible fBm with large Hurst index, admit long range repulsive forces
between the constituents.

It is proper to note that as far as more realistic polymer models are
concerned, these Gaussian models must be modified to suppress
self-crossings, see e.g. \cite{BBCES}\cite{GOSS}. 
%Extensions of these
%numerical and analytical studies are under way \cite{inprep}.

\section{The Nature of the Interactions along the Chain}

For applications in physics it would be desirable to rewrite the energy form 
$E(y)=\left( y,Ay\right) $ as an expression in terms of the\ monomer
position variables $x_{k\text{ }}$. Indeed one has

\begin{lemma}
	For $x=x_{0},...,x_{n}\in R^{n+1}$ and $y_{k}=x_{k\ }-x_{k-1},$ $A$ a
	symmetric $n\times n$ matrix, the identity%
	\begin{equation}
		\left( y,Ay\right) =\sum\limits_{k,l=1}^{n\
		}a_{kl}y_{k}y_{l}=\sum\limits_{k,l=0}^{n}g_{kl}(x_{k}-x_{l})^{2}  \label{A}
	\end{equation}%
	holds with 
	$$
		g_{kl}=-\frac{1}{2}\left( a_{kl}+a_{k+1,l+1}-a_{k,l+1}-a_{k+1,l}\right) ,
	$$%
	setting $a_{kl}=0$ whenever $k,l\notin $ $\left\{ 1,...,n\right\} .$
	
	Conversely, for $s\leq t$%
	$$
		a_{st}=2\sum_{i=0}^{s-1}\sum_{k=t}^{n}g_{ik.}
	$$
\end{lemma}

\begin{proof}
	By direct verification.
\end{proof}

\begin{remark}
	As indicated above, expressions such as $\left( y,Ay\right) $ occur in the
	probability densities of mean zero Gaussian processes. In physics they play
	the role of potential energies in canonical ensembles and the lemma shows
	that these energies can always be understood as the sum of harmonic
	oscillator potentials with strength $g_{kl}$ acting on the monomer pairs $%
	k,l $, see e.g. \cite{RC}.
\end{remark}

To understand the nature in particular of fBm models which were suggested
for chain polymers it is interesting to compute those coupling constants. We
show these for a chain of length $n=61$; specifically we display the
couplings of the center constituent to all the others. For small Hurst index
($H<1/2$), as one might expect from the more compact nature of the
corresponding fBm trajectories, all of the interactions are attractive.

\begin{figure}
	\includegraphics[width=\textwidth]{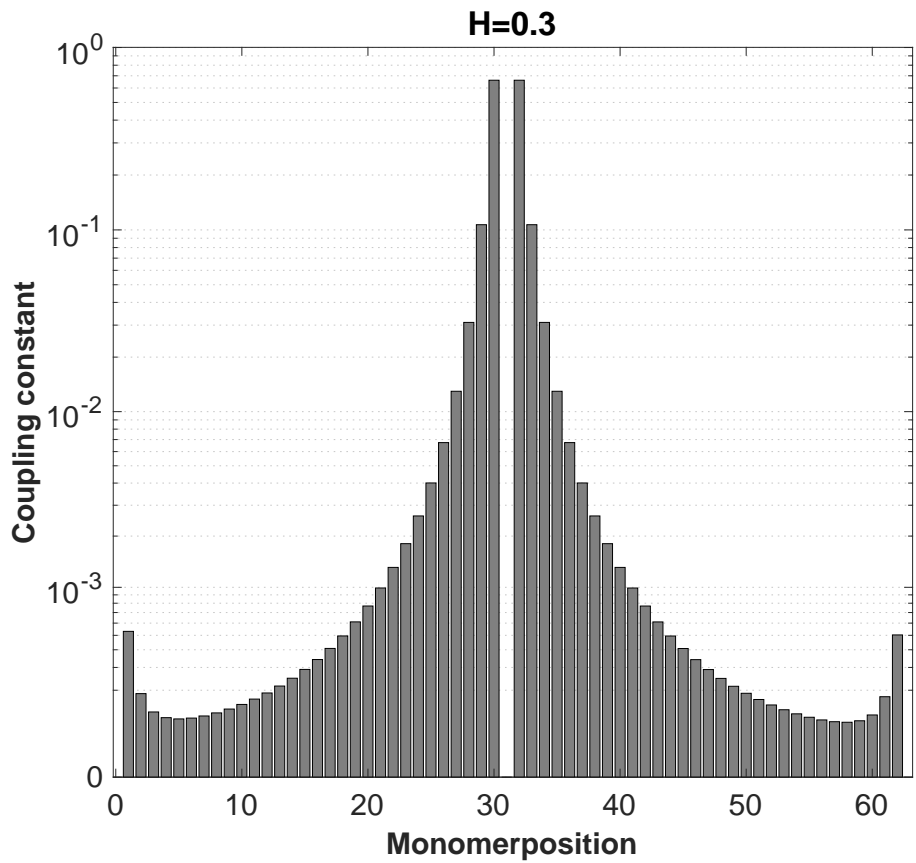}
\caption{Coupling constants for a fBm chain with 61 monomers and Hurst index H=0.3.}
\end{figure}

$$
	g_{31,i}>0\text{ \ for all }i\neq 31.
$$

The rise at both ends of the chain is a finite-length effect.

Conversely, for $H>1/2,$ nearest neighbors are again attracted, but
constituents further away on the chain are repelled.

\begin{figure}
\includegraphics[width=\textwidth]{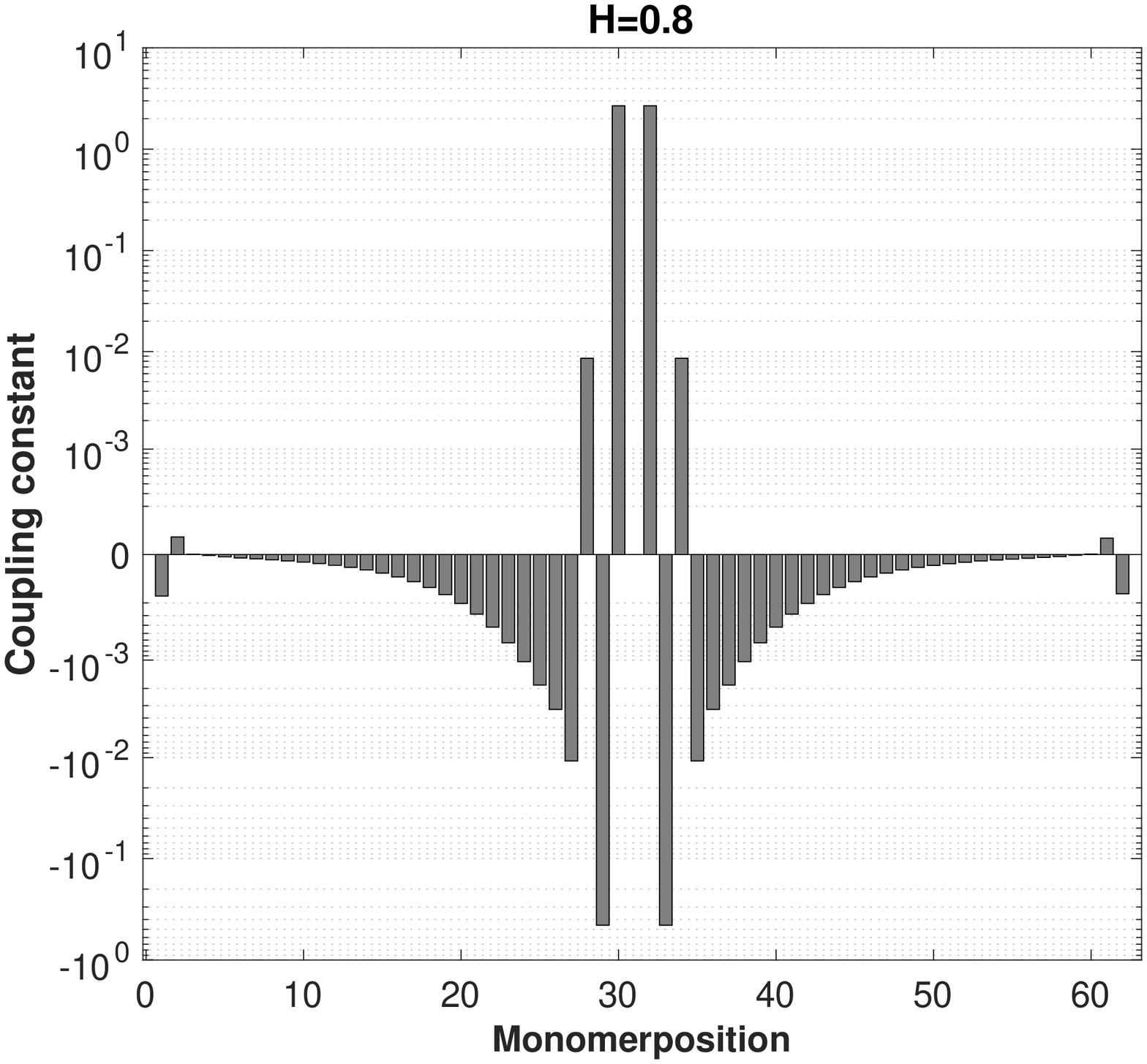}
\caption{Coupling constants for a fBm chain with 61 monomers and Hurst index H=0.8.}
\end{figure}
Finally, for\ $H\ \gtrapprox $ $0.75964,$ a value obtained via a bisection method, the third-nearest neighbors are also
attracted, in our example%
$$
	g_{31,i}>0\text{ \ for }i=31\pm 1\text{\ and for }i=31\pm 3
$$%
\begin{figure}\label{fig_critcoup}
	\includegraphics[width=\textwidth]{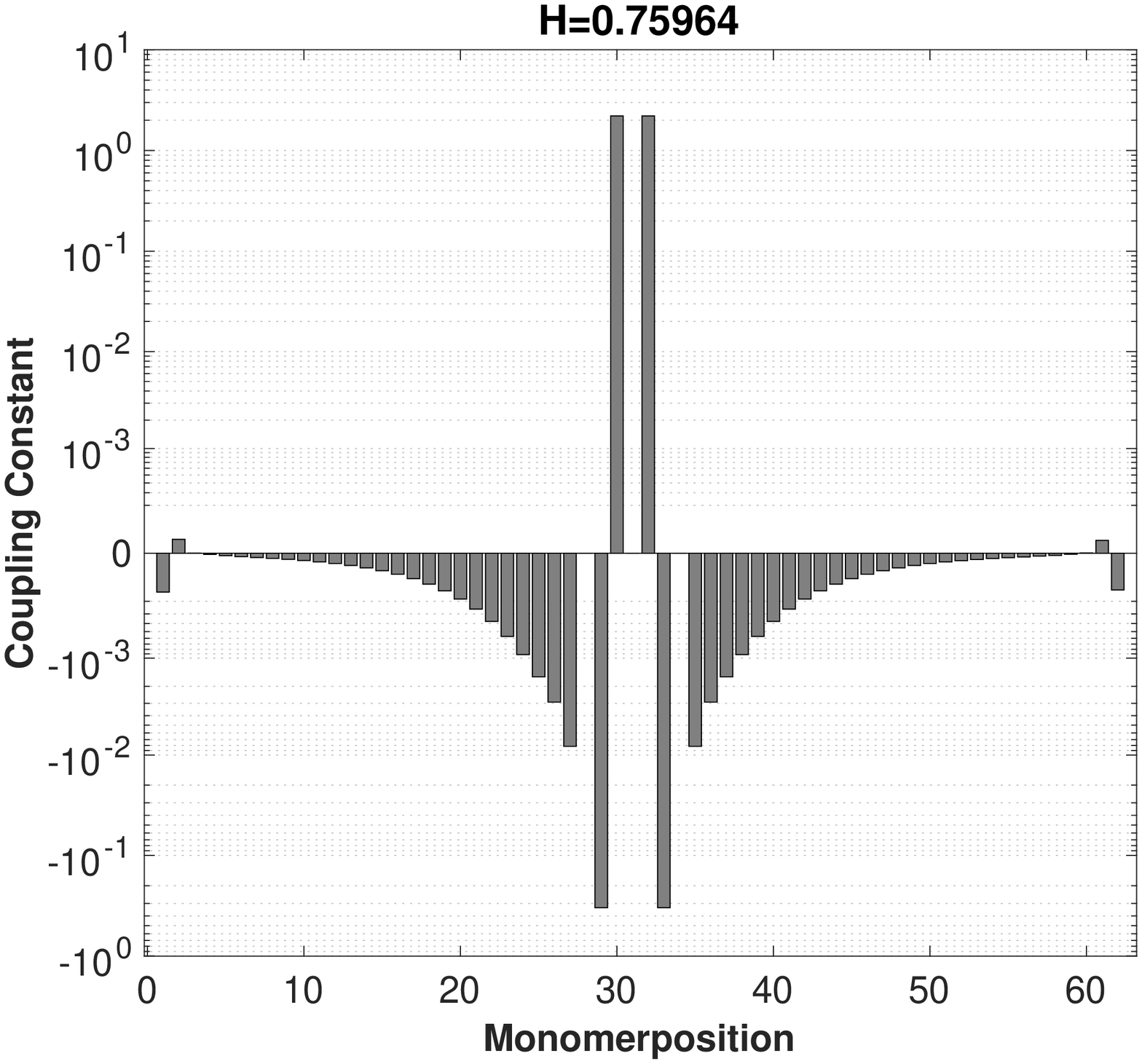}
\caption{Coupling constants for a  fBm chain with 61 monomers and Hurst index H=0.75964. The coupling constant of the third oscillators from the middle are zero.}
\end{figure}
\section{The Energy E as a Quadratic Form in x}

At times it can be useful to have the energy%
\begin{equation}
	E(x)=\sum\limits_{k>l=0}^{n}g_{kl}(x_{k}-x_{l})^{2}  \label{E(x)}
\end{equation}%
as a quadratic form%
$$
	E(x)=\left\langle x,\mathcal{H}x\right\rangle
$$%
with the matrix $\mathcal{H}$ to be determined. Setting%
$$
	g_{kl}=g_{lk\text{\ }}\text{ \ \ and \ }g_{kk}=0
$$%
one finds directly 
\begin{equation}
	h_{ik}=\left\{ 
	\begin{array}{ccc}
		-g_{ki} & \text{if} & k\neq i \\ 
		\sum_{j}g_{jk} & \text{if} & k=i%
	\end{array}%
	\right.  \label{H}
\end{equation}%
Note that the matrix $\mathcal{H}$ has $x=(1,....,1)$ as eigenvector with
eigenvalue zero, in accordance with equation (\ref{E(x)}). \ This
corresponds to the "increment" vector $y=(x_{k+1}-x_{k\ })_{k=0}^{n-1}$
being zero. All the other eigenvalues of $\mathcal{H}$ are then those of the
matrix $A$ in (\ref{A}).

\section{Cyclic Conformations}

In an equidistant periodic placement of $n~$\ indices $k=0,...,n-1$ on a
circle, there are $\left[ \frac{n}{2}\right] $ distinct geodesic distances%
$$
	d(i,k)=\min \left( \left\vert i-k\right\vert ,n-\left\vert i-k\right\vert
	\right) .
$$

\subsection{Periodic Fractional Brownian Motion}

As pointed out by \cite{Istas}, periodic fBm on a circle of length $L$ is
obtained by setting%
\begin{eqnarray}
	\mathbb{E}\left( \left( b^{H}(s)-b^{H}(t)\right) ^{2}\right) &=&\left( \min \left(
	\left\vert s-t\right\vert ,L-\left\vert s-t\right\vert \right) \right)
	^{2H} \nonumber\\
	&=:&\left( d(s,t)\right)^{2H}, \nonumber
\end{eqnarray}
which serves to define a periodic centered Gaussian process $b^{H}$ iff $%
H\leq 1/2.$ \ 

We put $b^{H}(0)=0$ without loss of generality, and the covariance is%
\begin{equation}
	\mathbb{E}\left( b^{H}(s)b^{H}(t)\right) =\frac{\left( \left( d(s,t)\right)
	^{2H}+\left( d(s,0)\right) ^{2H}-\left( d(s,t)\right) ^{2H}\right)}{2}
	\label{cov}
\end{equation}

As a discretized version we set $L=n+1$%
$$
	X_{k}=b^{H}(k),\text{ \ \ }k=0,1,\ldots ,n
$$%
with increments%
$$
	Y_{k}=X_{k\ +1}-X_{k},\text{ \ \ }k=0,...,n\ -1
$$%
and a joint density%
$$
	\rho (y)=const.\exp \left( -\frac{1}{2}(y,Ay)\right)
$$%
where the energy matrix%
$$
	A=R^{-1}
$$%
is obtained as the inverse of the correlation matrix%
$$
	R_{ik}=E(Y_{i}Y_{k}).
$$%
This covariance kernel is positive definite, if and only if the Hurst index $%
H$ is smaller than $1/2$; for $H>1/2\,\ $the matrices 
\begin{equation}
	\mathbb{E}\left( X_{k}X_{l}\right) =\frac{
	d^{2H}(k)+d^{2H}(l)-d^{2H}(k-l)}{2}  \label{cor}
\end{equation}%
$$
	\mathbb{E}\left( Y_{k}Y_{l}\right) =\frac{
	d^{2H}(k-l+1)+d^{2H}(k-l-1)-2d^{2H}(k-l)}{2}
$$%
will fail to be positive semi-definite.

The coupling constants $g_{k}$ for discrete periodic
fBm of the Hurst index $H=0.3$ is given in Figure 4 . 

\begin{figure}\label{cyc_03}
	\includegraphics[width=\textwidth]{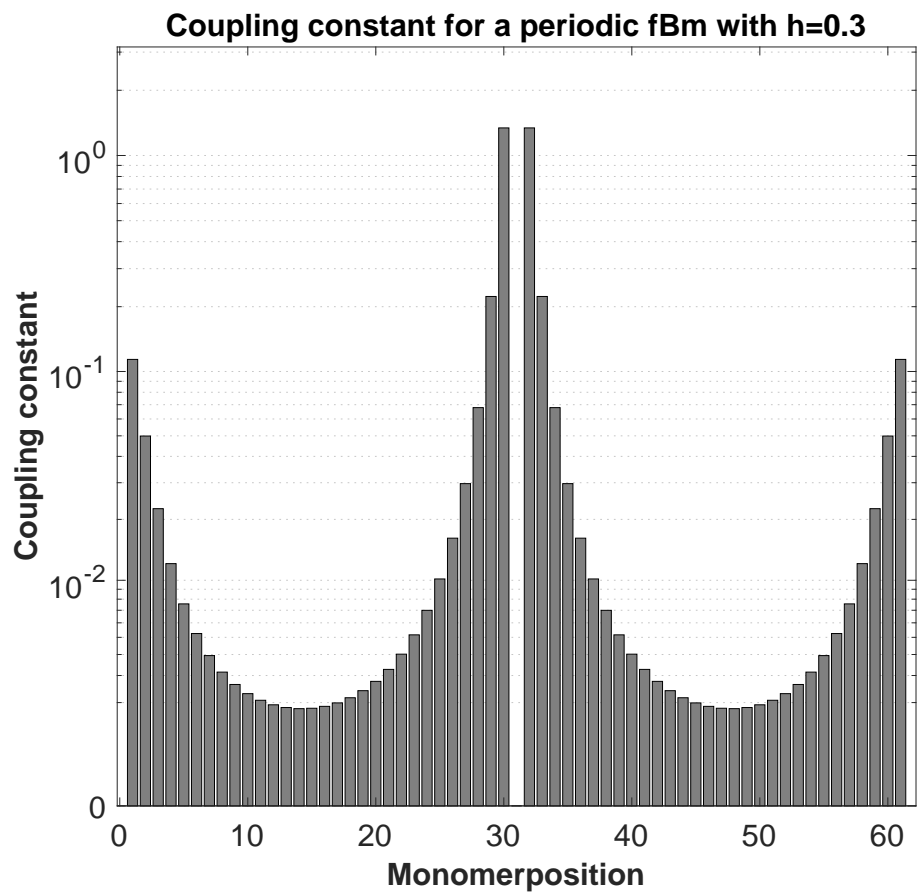}
	\caption{Coupling constants for a periodic fBm ring with 61 monomers and Hurst index H=0.3.}
\end{figure}

\subsection{General Cyclic Models}

For the $n$ Gaussian variables\ $y_{k}$ we define 
$$
	\rho (y)=const.\exp \left( -\frac{1}{2}(y,Ay)\right)
$$

starting from $\mathcal{H}(x)$ and the equation (\ref{A}), assuming that the 
$g_{ik}$ depend only on the "geodesic" distance between the indices $i$ and $%
k\,\ $i.e.

$$
	g_{ik}=g_{d(i,k)}.
$$

We introduce the symmetric $(n+1)\times (n+1)$ matrix $G=circ(0,g_{1},g_{2}%
\ldots ,g_{2},g_{1})$ 
$$
	G=\left( 
	\begin{array}{ccccccccc}
		0 & g_{1} & g_{2} & ... & ... & ... & g_{3} & g_{2} & g_{1} \\ 
		g_{1} & 0 & g_{1} & ... & ... & ... & ... & g_{3} & g_{2} \\ 
		g_{2} & g_{1} & 0 &  &  &  &  &  & g_{3} \\ 
		\vdots &  &  & \ddots &  &  &  &  & \vdots \\ 
		\vdots &  &  &  &  & \ddots &  &  & \vdots \\ 
		g_{3} &  &  &  &  &  & 0 &  & g_{2} \\ 
		g_{2} &  &  &  &  &  &  & 0 & g_{1} \\ 
		g_{1} & g_{2} & g_{3} & ... & ... & ... & g_{2} & g_{1} & 0%
	\end{array}%
	\right) .
$$

In view of (\ref{H}) $\mathcal{H}$ then is of the form%
\begin{equation}
	\mathcal{H}=g \boldsymbol{1-G}  \label{eitsch}
\end{equation}%
where $\boldsymbol{1}$ is the $\left( n+1\right) \times \left(
n+1\right) $ unit matrix, and $g$ is the row sum of $\boldsymbol{G}$ 
$$
	g=\sum_{i=1}^{n}g_{i}\ \text{\ \ with \ }g_{k}:=g_{n+1-k}\text{ \ if \ \ }%
	k>n/2.
$$

The energy spectrum is that of the matrix%
$$
	\mathcal{H}=g \boldsymbol{1-G}={circ}\left( \sum_{k=1}^{n\text{ }%
	}g_{k},-g_{1},\ldots ,-g_{n}\right)
$$%
with the exception of the eigenvalue zero. \ By formula (\ref{ev}) in the
Appendix, one has an explicit form for these energy eigenvalues as a
function of the coupling constants:%
\begin{equation}
	\lambda _{m}=\sum_{k=1}^{n}g_{k}\left( 1-\cos \left( \frac{2mk\pi }{n+1}%
	\right) \right) .  \label{lambda}
\end{equation}

\begin{example}
	\index{EX1|see{EX1}}Particular models of interest might be ones with a
	scaling behavior%
	$$
		g_{k}\sim k^{-\gamma }
	$$
	
	and \ 
	
	(1) all couplings\ positive (generalizing $H<1/2$), or
	
	(2) all couplings except the nearest neighbor one repulsive (as an
	alternative to a periodic fBm with $H>1/2$):%
	$$
		g_{1}>0,\ \ \ g_{k}<0\ \ for\ \ k>1.
	$$
	
	Models of this latter nature are expected to produce stiff polymers.
\end{example}

\subsection{Stiff Cyclic Polymers}

Stiff chain polymers might be modelled by fBm with Hurst parameter $H>1/2.$
For such models we saw that the next-to-nearest neighbor coupling constants
are negative, the nearest neighbor is attracted, such as e.g. in \cite%
{Winkler}. As Istas points out \cite{Istas}, there is no cyclic version for
fBm with $H>1/2$. \ But other cyclic models with attractive-repulsive
dynamics can be constructed, e.g. as follows.

\begin{example}\label{Exam:XXX}
	Taking%
	$$
		G=circ(0,g_{1},g_{2},0,\ldots ,0,g_{2},g_{1}),
	$$%
	hence 
	$$
		\mathcal{H}=circ(2\left( g_{1}+g_{2}\right) ,-g_{1},-g_{2},0,\ldots
		,0,-g_{2},-g_{1}),
	$$%
	for which, using the formula (\ref{ev}) in the Appendix, we obtain the
	energy eigenvalues%
	$$
		\lambda _{m}=2g_{1}\left( 1-\cos \left( 
		\frac{2m\pi }{n+1}\right) \right) +2g_{2}\left( 1-\cos \left( \frac{4m\pi }{%
			n+1}\right) \right) .
	$$%
	Apart from the spurious eigenvalue $\lambda _{0}=0,$ all eigenvalues are
	positive, i.e. this choice of coupling produces an admissible ring model if $%
	g_{2}$ $\geq -\frac{1}{4}g_{1}.$ The way to see this is to look at the
	function%
	$$
		f(x)=1-\cos \left( 2\pi x\right) +c\left( 1-\cos \left( 4\pi x\right) \right)
	$$%
	between the endpoints $x=0$ $\left( \text{and }x=1\text{ symmetrically}%
	\right) $ where, for $x\gtrsim 0,$ it is of the approximate form%
	$$
		f(x)\approx 2\pi ^{2}\left( 1+4c\right) x^{2}
	$$%
	i.e. for positivity 
	$$
		\ 4c=4\frac{g_{2}}{g_{1}}\geq -1
	$$%
	is required.
\end{example}

\begin{remark}
	Had we chosen $g_{k}$ instead of $g_{2}$ non-zero the corresponding
	restriction would be%
	$$
		\ \ k^{2}\frac{g_{k}}{g_{1}}\geq -1,
	$$%
	meaning that longer range repulsion can only occur with ever smaller
	coupling strength (recall that a Hooke's law type force becomes stronger
	with increasing distance, hence could then more easily produce instability
	if negative).
\end{remark}

To construct long range repulsive interactions one needs to ensure
positivity (except for the spurious zero eigenvalue associated with the
vector $x\ $with $x_{k}=c$ for all $k$).

\begin{example}
	In Example \ref{Exam:XXX} , setting $g_{1}>0$ and $g_{k}$ negative\ for $k>1$
	positivity holds for 
	\begin{equation}
		g_{1}>\pi ^{2}\sum_{k=2}^{n/2}k^{2}\left\vert g_{k}\right\vert .
		\label{bound}
	\end{equation}
	
	For this estimate, one uses$\ $(\ref{lambda}) for $n/2\geq i>0$%
	$$
		\lambda _{i}=2g_{1}\left( 1-\cos \left( \frac{2\pi i}{n+1}\right) \right)
		-2\sum_{k=1}^{n/2}g_{k}\left( 1-\cos \left( \frac{2\pi ik}{n+1}\right)
		\right)
	$$%
	and 
	$$
		\left( \frac{x}{\pi }\right) ^{2}\leq 1-\cos x\ \text{and}\ 1-\cos x\leq
		\left( \frac{x}{2}\right) ^{2}
	$$%
	as a lower bound for the first and an upper bound on the second term.
	
	To ensure the estimate (\ref{bound}) e.g. for $g_{k}=-ck^{-\gamma }$ and
	arbitrary $n$, $\ $it$\ $is thus sufficient to choose $\gamma >3,$and $c\ $
	such that%
	$$
		g_{1}>c\pi ^{2}\sum_{k=2}^{\infty }k^{-\gamma +2}\ =c\pi ^{2}\left( \zeta
		(\gamma +2\right) -1),
	$$%
	where $\zeta $ is Riemann's zeta function.
\end{example}

\section{A Remark on the Brownian Case}

Contrary to the statement in \cite{Istas}, the construction of periodic
(f)Bm \ is via a \textit{semi}-positive definite covariance matrix for the
Brownian case $H=1/2.$ \ 

\begin{example}
	Consider a Brownian ring with $n=6$ steps: $X=\left( b(0),...,b(6\right) ).$%
	One computes%
	$$
		cov_{Y}=circ(1,0,0,-1,0,0)=\left( 
		\begin{array}{cccccc}
			1 & 0 & 0 & -1 & 0 & 0 \\ 
			0 & 1 & 0 & 0 & -1 & 0 \\ 
			0 & 0 & 1 & 0 & 0 & -1 \\ 
			-1 & 0 & 0 & 1 & 0 & 0 \\ 
			0 & -1 & 0 & 0 & 1 & 0 \\ 
			0 & 0 & -1 & 0 & 0 & 1%
		\end{array}%
		\right) .
	$$%
	Eigenvectors are $\left( 1,0,0,1,0,0\right) $ with eigenvalue zero, and $%
	(-1,0,0,1,0,0)$ with eigenvalue 2, together with their cyclic permutations.
	I.e. contrary to the assertion of Istas, the covariance matrices are no more
	positive definite when $H=1/2$.
\end{example}

More generally, Istas \cite{Istas}, p. 257, finds%
$$
	E\left( \left\vert \int_{0}^{2\pi }b^{H}\left( t\right) \exp \left(
	int\right) dt\right\vert ^{2}\right) =-\frac{4\pi ^{2}}{n^{2H+1}}%
	\int_{0}^{\pi }x^{2H}\cos \left( nx\right) dx,
$$%
and in the Brownian case the rhs is positve for odd $n\,$, but%
$$
	\int_{0}^{\pi }x\cos \left( nx\right) dx=0\text{ \ for \ even }n>0,
$$%
$\ $contrary to the assertion of Istas, nevertheless \emph{semi}-positivity
holds.

\begin{proposition}
	For $H=1/2$ a version of \ the periodic Brownian motion $b(\cdot )$ is
	obtained in terms of the Wiener process $B(\cdot )$ if one sets%
	$$
		b(t):=\left\{ 
		\begin{array}{ccc}
			B(t) & \text{for} & 0\leq t\ \leq \pi \\ 
			B(\pi )-B(t-\pi ) & \text{for} & \pi \leq t\ \leq 2\pi%
		\end{array}%
		\right. .
	$$
\end{proposition}

\begin{proof}
	Setting e.g. $L=2\pi \ $\ and $0\leq s\leq t\leq 2\pi $ without loss of
	generality, the covariance function given in (\ref{cov}) for $H=1/2$ is%
	$$
		E\left( b^{H}(s)b^{H}(t)\right) =\left\{ 
		\begin{array}{ccc}
			s & \text{for} & 0\leq s\leq t\leq \pi \\ 
			\pi +s-t & \text{for} & 0\leq s\leq \pi \text{ }\leq t\text{ and }t-s\leq \pi
			\\ 
			0 & \text{for} & 0\leq s\leq \pi \text{ }\leq t\text{ and }t-s\geq \pi \\ 
			2\pi -t & \text{for} & \pi \leq s\leq t%
		\end{array}%
		\right. .
	$$%
	Since $b(\cdot )$ as defined in the proposition is centered Gaussian, it is
	sufficient - and straightforward - to verify that it has this same
	covariance function.
\end{proof}

It is worth pointing out that a Brownian bridge such as 
\begin{equation}
	b(t)=B(t)-\frac{t}{2\pi }B(2\pi )  \label{pBm}
\end{equation}%
would not lead to stationary increments on the circle.\cite{Toeplitz}

\section{Supplementary Material on Circulant Matrices}

Note that Istas' cyclic correlation matrix is circulant and symmetric, hence
also its inverse.

Circulant matrices are of the general form%
$$
	C=\left( 
	\begin{array}{ccccc}
		c_{0} & c_{1} & c_{2} & \cdots & c_{n-1} \\ 
		c_{n-1} & c_{0} & c_{1} & \ddots & \vdots \\ 
		\vdots & c_{n-1} & \ddots & \ddots \  & \vdots \\ 
		\vdots &  & \ddots & \  & c_{1} \\ 
		c_{1} & \cdots & \cdots & c_{n-1} & c_{0}%
	\end{array}%
	\right) .
$$

Note that a real circulant matrix is symmetric iff 
$$
	c_{k}=c_{n-k}\text{ \ for \ \ }k>0.
$$%
Normalized eigenvectors \emph{for any }$C$ are%
$$
	e_{m}=\frac{1}{\sqrt{n}}\left( 1,\varrho ,\ldots ,\varrho ^{n-1}\right) 
	\text{ \ \ with \ \ }\varrho =e^{-\frac{2\pi i}{n}},
$$%
the corresponding eigenvalues are 
$$
	\lambda _{m}=\sum_{k=0}^{n-1}c_{k}e^{-\frac{2ikm\pi }{n}}=%
	\sum_{k=0}^{n-1}c_{k}\varrho ^{km},
$$%
i.e. they are simply the discrete Fourier transform of the first row of $C$.
For symmetric $C$ this formula gives%
\begin{eqnarray}
	\lambda _{m} &=&c_{0}+\sum_{k=1}^{n-1}c_{k}e^{-\frac{2ikm\pi }{n}%
	}=c_{0}+\sum_{k=1}^{n-1}c_{n-k}e^{-\frac{2imk\pi }{n}}\  \\
	&=&c_{0}+\frac{1}{2}\left( \sum_{k=1}^{n-1}c_{k}e^{-\frac{2imk\pi }{n}%
	}+\sum_{k=1}^{n-1}c_{k}e^{\frac{2imk\pi }{n}}\right)  \nonumber \\
	&=&\sum_{k=0}^{n-1}c_{k}\cos \left( \frac{2mk\pi }{n}\right) \text{ \ for \ }%
	m=0,\ldots ,n-1.  \label{ev}
\end{eqnarray}%
These are degenerate:%
$$
	\lambda _{m}=\lambda _{n-m}\text{ \ \ for \ }m>0
$$%
so that the corresponding eigenvectors can then be combined to be
real-valued. These $n\ $ eigenvectors are 
$$
	u_{j}^{(m)}=\cos \left( \frac{2\pi jm}{n}\right) ,\text{ \ \ }m=0,\ldots ,%
	\left[ \frac{n}{2}\right]
$$%
and%
$$
	v_{j}^{(m)}=\sin \left( \frac{2\pi jm}{n}\right) ,\text{ \ \ }m=1,\ldots ,%
	\left[ \frac{n-1}{2}\right] .
$$

\section{Conclusions and Outlook}

Having clarified the dynamical nature of fBm models for chains and rings,
one will next implement excluded volume effects. Where fBm rings fail to
exist we indicate how the control of energy eigenvalues allows for the
construction of alternate stiff ring models with a short or long range
interaction along the chain.

\begin{acknowledgement}
	We are most grateful to Roland Winkler for inspiration and many helpful
	discussions.
\end{acknowledgement}

\end{document}